\documentclass[11pt,letterpaper]{llncs}

\usepackage{float}
\usepackage{fullpage}
\usepackage{graphicx}
\usepackage{color}
\usepackage{amsmath}
\usepackage[sort,nocompress]{cite}
\usepackage{caption}
\usepackage{mathtools}

\newcommand{\bigO}{\mathcal{O}}
\newcommand{\BWT}{\ensuremath{\mathit{BWT}}}

\begin{document}

\title{String Attractors}
\author{Nicola Prezza}
\institute{Technical University of Denmark, DTU Compute\\
	\texttt{npre@dtu.dk}}

\maketitle

\begin{abstract}
Let $S$ be a string of length $n$. In this paper we introduce the notion of \emph{string attractor}: a subset of the string's positions $[1,n]$ such that every distinct substring of $S$ has an occurrence crossing one of the attractor's elements. 
We first show that the minimum attractor's size yields upper-bounds to the string's repetitiveness as measured by its linguistic complexity and by the length of its longest repeated substring.
We then prove that all known compressors for repetitive strings induce a string attractor whose size is bounded by their associated repetitiveness measure, and can therefore be considered as approximations of the smallest one. 
Using further reductions, we derive the approximation ratios of these compressors with respect to the smallest attractor and solve several open problems related to the asymptotic relations between repetitiveness measures (in particular, between the the sizes of the Lempel-Ziv factorization, the run-length Burrows-Wheeler transform, the smallest grammar, and the smallest macro scheme). These reductions directly provide approximation algorithms for the smallest string attractor.
We then apply string attractors to solve efficiently a fundamental problem in the field of compressed computation: we present a universal compressed data structure for text extraction 
that improves existing strategies simultaneously for \emph{all} known dictionary compressors and that, by recent lower bounds, almost matches the optimal running time within the resulting space. 
To conclude, we consider generalizations of string attractors to labeled graphs, show that the attractor problem is NP-complete on trees, and provide a logarithmic approximation computable in polynomial time. 
\end{abstract}

\section{Introduction}

Being able to determine the level of repetitiveness of a string is of critical importance in areas such as computational biology, data compression, stringology, and cryptography. In biology, genome repetitions often  indicate the presence of transposable elements and represent an important source of genetic variation. In data compression and stringology, repetitions are exploited to represent the string in a more compact form and to study properties such as periodicity and maximal repeats. In cryptography, repetitions in pseudo-random generators or in the encrypted messages should be avoided as they can produce easier-to-break codes.  Typical definitions of repetitiveness either look at the string's composition in terms of distinct $k$-mers or rely on the output of particular compressors. An example of the former case is  \emph{linguistic complexity}~\cite{trifonov1990making,troyanskaya2002sequence}: the rate between the number of distinct substrings and the maximum number of distinct substrings that can appear in a string of the same length on the same alphabet. The linguistic complexity of a string is of relevance in computational biology, where is often used to determine the complexity of a genome~\cite{trifonov1990making,troyanskaya2002sequence}. Despite intuitively capturing the degree of repetitiveness of a string, however, linguistic complexity does not tell us anything about how repetitions could be used to remove redundancy from the string. More empirical definitions---mostly used in the realms of data compression and stringology---define a string to be repetitive when its compressed representation is small when compared to the string's original size. Despite being intuitive (although somewhat circular), this definition presents an important issue: the landscape of data compression is composed of myriads of different compressors for repetitive strings---each coming with its distinct repetitiveness measure---so it is not always clear which one should be used given a particular string.
One effective compression strategy of this kind is, for example, to build a context-free grammar that generates only the string. The grammar takes the name of \emph{Straight Line Program} (SLP) if all rules have the form $X\rightarrow AB$, or \emph{run-length SLP} (RLSLP) if also rules of the form $X \rightarrow A^\ell$ are allowed, for any $\ell > 2$.
The problems of finding the smallest SLP---of size $g^*$---and the smallest run-length SLP---of size $g_{rl}^*$---are NP-complete~\cite{hucke2016smallest}, but fast and effective approximation algorithms are known, e.g. LZ78~\cite{ziv1978compression}, LZW~\cite{welch1984technique}, Re-Pair~\cite{larsson2000off}, Bisection~\cite{kieffer2000universal}. Another strategy, more powerful than grammar compression, is that of replacing repetitions with pointers to other locations in the string. The most powerful and general scheme falling into this category takes the name of \emph{pointer macro scheme}~\cite{storer1982data}, and consists in a set of substring equalities that allow to unambiguously reconstruct the string. Finding the smallest such system---of size $b^*$---is also NP-complete~\cite{gallant1982string}. However, if we add the constraint of unidirectionality  (i.e. text can only be copied from previous positions), then the optimal such scheme~\cite{lempel1976complexity} of size $z$ can be computed in linear time~\cite{crochemore2008computing} with a greedy algorithm known as LZ77.
Other effective techniques to compress repetitive strings include the run-length Burrows-Wheeler transform~\cite{burrows1994block} (RLBWT) and the Compact Directed Acyclic Word Graph~\cite{blumer1987complete,crochemore1997direct} (CDAWG). With the first technique, we sort all circular string permutations in a $n\times n$ matrix; the BWT is the last column of this matrix. The BWT contains few equal-letter runs if the string is very repetitive, therefore run-length compression often significantly reduces the size of this string permutation~\cite{makinen2010storage}. The number $r$ of runs in the BWT is yet another good measure of repetitiveness. Finally, one could build an automaton recognizing the string's suffixes, and indicate with $e$ its number of edges. Also the size $e^*$ of the smallest such automaton---the CDAWG---grows sublinearly with $n$ if the string is very repetitive~\cite{belazzougui2015composite}. Both RLBWT and CDAWG can be computed in linear time~\cite{munro2016space,belazzougui2014linear,nong2009linear,crochemore1997compact}. 
Few relations are known among these measures. Rytter~\cite{rytter2003application} proved that $g^* \in \bigO(z\log(n/z))$ by showing a reduction from grammars to unidirectional parses, and Belazzogui et al.~\cite{belazzougui2015composite,belazzougui2017representing} showed that $g^*,z,r \in \bigO(e^*)$.
%and $g^*_{rl} \in \bigO(r \log(n/r))$. 
In practice, $e^*$ is much larger (by about one order of magnitude) than the other repetitiveness measures~\cite{belazzougui2015composite}.

At this point, it is natural to ask the following questions: what common principle---if any---stands at the basis of these compressors? how is the compressed string's size related with the string's linguistic complexity?
In this paper, we propose an answer to these questions. We introduce the notion of \emph{string attractor}: a set of the string's positions such that all distinct  substrings have an occurrence crossing one of the attractor's elements. We show an upper-bound to the string's linguistic complexity and a lower bound on the length of the longest repeated substring as a function of the smallest attractor size for that string: in particular, a smaller attractor implies a more repetitive string. We furthermore prove that these combinatorial objects are a natural generalization of all known compressors for repetitive strings: grammars, LZ77, macro schemes, RLBWT, and automata recognizing the string's suffixes (e.g. the CDAWG and  the suffix tree) induce string attractors whose sizes are bounded by their associated repetitiveness measures, and  can therefore be interpreted as approximations of the smallest one. We provide techniques to derive a macro scheme and a grammar from any string attractor, and use these reductions to derive the approximation rates of the above compressors with respect to the smallest string attractor. Using the same techniques, we moreover uncover new relations between repetitiveness measures: we show that $g^*,z \in \bigO(b^*\log^2(n/b^*))$ and $g^*,z \in \bigO(r\log^2(n/r))$, thus providing the first bounds relating these measures and solving an open problem explicitly stated in~\cite{prezza2016can}.

After introducing and studying the properties of string attractors, we present an application of these combinatorial objects to the domain of compressed data structures. 
We start by noticing that our conversion technique yielding a macro scheme implies that only a logarithmic blow-up in space is needed to reconstruct the string from a string attractor. Indeed, we show that this property can be turned into a very efficient data structure---we call it \emph{A-DAG} (Attractor Directed Acyclic Graph)---for text extraction. Thanks to the universality of string attractors, this technique can be used \emph{verbatim} to extract text from any compressed representation within a space bounded by a (small) function of its associated repetitiveness measure. 
Let $\gamma$ be the size of a string attractor of a length-$n$ string $T$ over an alphabet $[1\dots,\sigma]$.
For any $\tau>0$, the A-DAG takes $\bigO(\gamma \tau \log_\tau(n/\gamma))$ words of space and supports the extraction of any length-$\ell$ substring of $T$ in $\bigO(\log_\tau(n/\gamma) + \ell\log\sigma/w)$ time. 
On one side, we obtain the first data structure supporting efficient random access on general macro schemes. In general, the A-DAG is faster than existing text-extraction data structures for all known repetitiveness measures (by our reductions, it is sufficient to replace $\gamma$ with $z$, $b$, $g$, $g_{rl}$, $r$, or $e$ in the above bounds). Moreover, a recent lower bound~\cite{CVY13} implies that our query times for $\tau=\log^\epsilon (n/\gamma)$ (for any small $\epsilon>0$) are very close to the optimum within the resulting space. 
Another interesting tradeoff is achieved choosing $\tau = (n/\gamma)^{\epsilon}$ for any constant $\epsilon>0$: in this case, we obtain optimal $\bigO(\ell\log(\sigma)/w)$ extraction time in the packed setting within $\bigO(\gamma^{1-\epsilon}n^\epsilon)$ space. This space is a weighted geometric average between the compressed and plain string sizes, is always at most $\bigO(n)$, and is $o(n)$ if $\gamma=o(n)$, i.e. if the string is asymptotically compressible.

In our last section we propose a  generalization of string attractors to labeled directed graphs capturing all distinct paths (interpreted as strings) in the graph. We show that the decisional version of the path attractor problem is NP-complete on labeled trees, and provide a polynomial-time logarithmic approximation based on a reduction to the set-cover problem.  
We conclude the paper by proposing several exciting open problems that naturally follow from our work, and discussing new research directions that could lead to a better understanding of the notion of repetitiveness.

Proof omitted within the first 12 pages are reported in Appendix \ref{sec:app} due to space limitations.

\section{Preliminaries}\label{sec:preliminaries}

When $s$ is a string, the notation $\overleftarrow s$ denote $s$ reversed.
For space reasons, we assume the reader to be familiar with the notions of Lempel-Ziv factorization~\cite{lempel1976complexity}, Burrows-Wheeler transform~\cite{burrows1994block}, run-length encoding, Compact Directed Acyclic Word Graph~\cite{blumer1987complete,crochemore1997direct} (CDAWG), and grammar compression~\cite{hucke2016smallest}.
In this paper we consider the version of the Lempel-Ziv factorization~\cite{lempel1976complexity} (LZ77 for short) without self-references and without trailing characters.

A \emph{macro scheme}\cite{storer1982data} is a set of $b$ directives of two possible types:
	\begin{itemize}
		\item $T[i..j] \leftarrow T[i'..j']$ (i.e. copy $T[i'..j']$ in $T[i..j]$), or
		\item $T[i] \leftarrow c$, with $c\in\Sigma$ (i.e. assign character $c$ to $T[i]$).
	\end{itemize}
such that the text can be unambiguously reconstructed from the directives. 

A \emph{bidirectional parse} is a macro scheme where the left-hand sides of the directives induce a text factorization, i.e. they cover the whole $T$ and they do not overlap. 
Note that LZ77 is a particular case of bidirectional parse (in particular, it is the optimal unidirectional one), and therefore is also a particular macro scheme. 

The \emph{height} of a macro scheme is defined as follows.  Let $i$ be a string position. If $i$ is an explicitly stored character, we define the \emph{character height} $h_i$ of $i$ to be equal to 1. Otherwise, let $j$ be the position from where $T[i]$ is copied according to the macro scheme. Then, we define the character height $h_i$ of $i$ to be $h_i = h_j+1$.
The macro scheme height $h$ is defined as the maximum character height $h = \max_{1\leq i\leq n}h_i$.

\section{String attractors}\label{sec:SA}

A string $S[1..n]$ is considered to be repetitive when the cardinality of the set $SUB_S = \{ S[i..j]\ |\ 1\leq i\leq j \leq |S|  \}$ of its distinct substrings is much smaller than the maximum number of distinct substrings on a text of the same length on the same alphabet. In the domain of data compression, a more relevant quantification $\gamma$ of the string's repetitiveness should, however,  also provide a meaningful way of representing $SUB_S$ within a space proportional to $\gamma$ so that $S$ itself can be reconstructed from this representation. In this respect, $SUB_S$ itself provides the ``ideal'' quantification of repetitiveness, but it is useless in practice as its size could be quadratic in $S$. As a result of these requirements, typical definitions of repetitiveness tend to be circular in the sense that rely on particular compressors such as the Lempel-Ziv factorization, straight-line programs, or the run-length Burrows-Wheeler transform. In other words, a string $S$ is usually considered to be repetitive when the size of its compressed representation is small when compared to the worst-case entropy  of the set of all length-$n$ strings on the same alphabet ($n\log\sigma$ bits). In this section, we tackle the problem from the opposite direction: we first give a more ``principled'' definition of string repetitiveness that naturally follows from the definition of $SUB_S$, and show that it directly relates with the number of distinct substrings. We then show that all known dictionary compressors are a particular case of our notion of repetitiveness, and show how our perspective of the problem can be exploited to uncover new properties of dictionary compressors and solve more efficiently data-structure problems related to them.
Our idea is the following: we fix a set of $\gamma$ positions on $S$ so that every element of $SUB_S$ has an occurrence crossing one of these positions, and claim that the size $\gamma^*$ of the smallest such set is a meaningful notion of string repetitiveness. We call the set of such positions \emph{string attractor}:

\begin{definition}\label{def: string attractor}
	A \emph{string attractor} of a string $T\in\Sigma^n$ is a set of $\gamma$ positions $\Gamma = \{j_1, ..., j_\gamma\}$ such that every substring $T[i..j]$ has an occurrence $T[i'..j'] = T[i..j]$ with $j_k \in [i',j']$, for some $j_k\in\Gamma$. 
\end{definition}

\begin{example}
	Note that $\{1, 2, .., n\}$ is always a string attractor (the biggest one) for any string. Note also that this is the only possible string attractor for a string composed of $n$ distinct characters.
\end{example}

\begin{example}\label{ex:attr}
	Consider the following string $S$, where we underlined the positions of a smallest string attractor $\Gamma^* = \{4,7,11,12\}$ of $S$.
	\begin{center}
		\texttt{CDA\underline BCC\underline DABC\underline C\underline A}
	\end{center}
	To see that $\Gamma^*$ is a valid attractor, note that every substring between attractor's positions has an occurrence crossing an attractor position (these substrings are \texttt{A,B,C,D,CD,DA,AB,CC}).	The remaining substrings cross an attractor position by definition.
	To see that $\Gamma^*$ is the smallest attractor for $S$, note that the alphabet size is $\sigma = 4 = |\Gamma^*|$, and any attractor $\Gamma$ must satisfy $|\Gamma| \geq \sigma$. 
\end{example}

Clearly, the set $\Gamma$ alone is not sufficient to reconstruct the original string as we do not even store characters associated with those positions. However, in Theorem \ref{th: attractor -> macro} we will show that only a logarithmic blow-up in space is needed to achieve this task. For now, we stress out that the size $\gamma^*$ of the smallest string attractor and the string's repetitiveness are two closely-related concepts. %First, we prove a lower bound on the size of $SUB_S$.
%\begin{lemma}\label{lem:SUB}
%	$|SUB_S|\geq \gamma^*$
%\end{lemma}
%\begin{proof}
%	Let $\Gamma^*$ be a minimum attractor for $S$, of size $\gamma^*$.	Define a relation $\phi:SUB_S \rightarrow \Gamma^*$ mapping each substring $s$ to the leftmost attractor crossing an occurrence of $s$. Since $\Gamma^*$ is an attractor for $S$, every substring is mapped to exactly one of $\Gamma^*$'s elements by $\phi$, which implies that $\phi$ is a function. Since  $\Gamma^*$ is the minimum attractor, each of its elements $x\in \Gamma^*$ is reached by at least one domain element (otherwise, $\Gamma^*-\{x\}$ would also be a valid attractor), which implies that $\phi$ is a surjective function, i.e. the thesis. \qed
%\end{proof}
%We now prove an upper-bound. 
First, we show an upper-bound to the number of distinct $k$-mers.

\begin{lemma}\label{lem:kmers}
	Let $\Gamma$ be a string attractor of size $\gamma$ for $S$. Then, $S$ contains at most $\gamma k$ distinct $k$-mers, for every $1\leq k \leq |S|$.
\end{lemma}
\begin{proof}
	By the attractor's definition, every distinct $k$-mer
	appearing in the string has an occurrence crossing some attractor's position. It follows that, in order to
	count the number of distinct $k$-mers, we can restrict our attention to the
	regions of size $2k-1$ overlapping the attractor's positions. The upper-bound $\gamma k$ easily follows. \qed
\end{proof}

%The linguistic complexity $LC(S)$ of a string $S\in [1,\sigma]^n$, as defined in \cite{troyanskaya2002sequence}, is the rate between $|SUB_S|$ and the maximum number $\sum_{k=1}^n\min\{\sigma^k, n-k+1\}$  of distinct substrings that could appear in a string from $[1,\sigma]^n$. $LC(S)$ is often used as an index of string repetitiveness. Using Lemma \ref{th:LR} we now give an upper-bound to $LC(S)$ as a function of $\gamma^*$.
%From Lemma \ref{lem:kmers}, we obtain the upper-bound $|SUB_S| \leq \sum_{k=1}^n\min\{\sigma^k, n-k+1, \gamma^*k\}$ to the size of $|SUB_S|$, which directly yields an upper-bound for $LC(S)$ in terms of $\gamma^*$:

From the above lemma, we directly obtain:

\begin{theorem}\label{th:LC}
	$|SUB_S| \leq \sum_{k=1}^n\min\{\sigma^k, n-k+1, \gamma^*k\}$
\end{theorem}
%\begin{proof}
%	Clearly, there cannot be more than $\min\{\sigma^k, n-k+1\}$ distinct $k$-mers in $S$. Combining these bounds with Lemmas \ref{lem:SUB} and \ref{lem:kmers}, we obtain our claim.\qed
%\end{proof}

The linguistic complexity $LC(S)$ of a string $S\in [1,\sigma]^n$, as defined in \cite{troyanskaya2002sequence}, is the rate between $|SUB_S|$ and the maximum number $\sum_{k=1}^n\min\{\sigma^k, n-k+1\}$  of distinct substrings that could appear in a string from $[1,\sigma]^n$. $LC(S)$ is often used as an index of string repetitiveness. Theorem \ref{th:LC} yields an upper- bound for $LC(S)$ as a function of $\gamma^*$.

\begin{example}
	For the string $S$ in example \ref{ex:attr}, we have $n=12$, $\sigma=4$, $\gamma^*=4$, and $|SUB_S|=57$. Then, $LC(S) \approx 0.814$ and Theorem \ref{th:LC} gives the bound $LC(S)\leq 0.9572$.
\end{example}

The length $\ell_{max}$ of the longest repeated substring is also indicative of the string's repetitiveness. In the following lemma we show that a smaller minimum attractor implies a larger $\ell_{max}$. Equivalently, we give a lower-bound on the size of the minimum attractor as a function of $\ell_{max}$.

\begin{lemma}\label{th:LR}
	Let $S \in [1,\sigma]^n$ be a string with minimum attractor's size $\gamma^*$ and longest repeated substring's length equal to $\ell_{max}$. Then, it holds $\ell_{max} \geq \frac{n-\gamma^*}{\gamma^*+1}$ and, equivalently, $\gamma^* \geq \frac{n-\ell_{max}}{\ell_{max}+1}$.
\end{lemma}
\begin{proof}
	Consider the maximum distance $j-i$ between any two positions $i,j$ in the smallest attractor for $S$. Then, by Definition \ref{def: string attractor}, the substring $S[i+1..j-1]$ of length $j-i-1$ between these two positions is not covered by any attractor's element. This means that $S[i+1..j-1]$ must have another occurrence crossing some other attractor element, which implies that it is repeated. The smallest maximum distance between the $\gamma^*$ attractor's elements is obtained when they are equally spaced. The claim easily follows after observing that the repeated substring should not cross any attractor's element ($n-\gamma^*$ positions are therefore left), and the $\gamma^*$ elements break the text in $\gamma^*+1$ factors. \qed
\end{proof}

\subsection{Reductions from compressors}\label{sec:compressors->SA}

In the next theorems we show reductions from all known compressors for repetitive strings to string attractors.
We start with a well-known lemma that will be used later:

\begin{lemma}\label{lem: grammars}
	Let $G = \{X_i \rightarrow A_iB_i,\ i=1,...,g'\} \cup \{Y_i \rightarrow Z_i^{\ell_i},\ \ell_i\geq 2,\  i=1,...,g''\}$ be a run-length grammar generating $T\in\Sigma^n$. For any substring $T[i..j]$ one of the following is true:
	\begin{enumerate}
		\item there exists a rule $X_k \rightarrow
		A_kB_k$ such that $T[i..j]$ is composed of a non-empty suffix of the expansion of $A_k$ followed by a non-empty prefix of the expansion of $B_k$, or
		\item there exists a rule $Y_k \rightarrow
		Z_k^{\ell_k}$ such that $T[i..j]$ is composed of a non-empty suffix of the expansion of $Z_k$ followed by a non-empty prefix of the expansion of $Z_k^{\ell_k-1}$
	\end{enumerate}
	Note that if the grammar does not allow run-length rules, then (1) must always be true for any substring.
\end{lemma}

\begin{theorem}\label{th: attractor g}
	Let $G = \{X_i \rightarrow A_iB_i,\ i=1,...,g'\} \cup \{Y_i \rightarrow Z_i^{\ell_i},\ \ell_i\geq 2,\  i=1,...,g''\}$ be a run-length grammar of size $g=g'+g''$ generating $T$. $G$ induces a family of string attractors, all of size $g$.
\end{theorem}
\begin{proof}
	Start with an empty string attractor $\Gamma_G = \emptyset$, and repeat the following  for $k=1,...,g'$. Choose any of the expansions $T[i..j]$ of $X_k$. By the production $X_k \rightarrow A_kB_k$, $T[i..j]$ can be factored as $T[i..j] = T[i..i']T[i'+1..j]$, where $T[i..i']$ and $T[i'+1..j]$ are expansions of $A_k$ and $B_k$, respectively. Insert position $i'$ in $\Gamma_G$.
	
	Now, repeat the following  for $k=1,...,g''$.
	Choose any of the expansions $T[i..j]$ of $Y_k$. By the production $Y_k \rightarrow Z_k^{\ell_k}$, $T[i..j]$ can be factored as $T[i..j] = T[i..i']T[i'+1..j]$, where $T[i..i']$ and $T[i'+1..j]$ are expansions of $Z_k$ and $Z_k^{\ell_k-1}$, respectively. Insert position $i'$ in $\Gamma_G$.

	To see that $\Gamma_G$ is a valid string attractor of size $g$, consider any  substring $T[i..j]$. By Lemma \ref{lem: grammars}, $T[i..j]$ has an occurrence $T[i'..j']$ spanning the expansion of some $A_k|B_k$ or of some $Z_k|Z_k^{\ell_k-1}$, where we highlighted with a vertical bar the crossing point. By the way we constructed $\Gamma_G$, this prefix/suffix decomposition of $T[i'..j']$ crosses (with the same split) one of the elements in $\Gamma_G$ (i.e. the one associated with production  $X_k \rightarrow
	 A_kB_k$ or $Y_k \rightarrow Z_k^{\ell_k}$). \qed
\end{proof}

Clearly, Theorem \ref{th: attractor g} applies also to grammars not allowing run-length rules: in this case, $g''=0$ and $g=g'$. Note that the theorem captures also LZ78, Lempel-Ziv Welch (LZW), and run-length encoding (since they are just particular grammars). We now show that also macro schemes are particular string attractors. 

\begin{theorem}\label{th: attractor MS}
	Let MS be a macro scheme of size $b$ of $T$. MS induces a string attractor $\Gamma_{MS}$ of size at most $2b$.
\end{theorem}
\begin{proof}
	
	Let $T[i_{k_1}..j_{k_1}] \leftarrow T[i'_{k_1}..j'_{k_1}],\ T[q_{k_2}] \leftarrow c_{k_2}$, with $1\leq k_1 \leq b_1$, $1\leq k_2 \leq b_2$, and $b=b_1+b_2$ be the $b$ directives of our macro scheme MS.
	We claim that $\Gamma_{MS} = \{i_1, \dots, i_{b_1}, j_1, \dots, j_{b_1}, q_1, \dots, q_{b_2}\}$ is a valid string attractor for $T$.
	
	%Let $T = B_1B_2...B_b$ be the factorization associated with MS, where each $B_i$ is either an explicitly stored character or a substring. In the latter case, $B_i$ is associated with a source $T[t..t+|B_i|-1] = B_i$. 
	Let $T[i..j]$ be any substring. All we need to show is that $T[i..j]$ has a \emph{primary occurrence}, i.e. an occurrence containing one of the positions $i_{k_1}$, $j_{k_1}$ or $q_{k_2}$.
	Let $i_1=i$ and $j_1 = j$. Consider all possible chains of copies (following the macro scheme directives) $T[i_1..j_1] \leftarrow T[i_2..j_2] \leftarrow T[i_3..j_3] \leftarrow ...$ that either end in some primary occurrence $T[i_k..j_k]$ or are infinite (note that there could exist multiple chains of this kind since the left-hand side terms of some macro scheme's directives could overlap).
	Our goal is to show that there must exist at least one finite such chain, i.e. that ends in a primary occurrence. Pick any $i_1\leq p_1 \leq j_1$. Since ours is a valid macro scheme, then $T[p_1]$ can be retrieved from the scheme, i.e. the directives induce a finite chain of copies  $T[p_1] \leftarrow  ... \leftarrow T[p_{k'}] \leftarrow c$, for some $k'$, such that $T[p_{k'}] \leftarrow c$ is one of the macro scheme's directives. We now show how to build a finite chain of copies $T[i_1..j_1] \leftarrow T[i_2..j_2] \leftarrow ... \leftarrow T[i_k..j_k]$ ending in a primary occurrence $T[i_k..j_k]$ of $T[i_1..j_1]$, with $k \leq k'$. By definition, the assignment $T[p_1] \leftarrow T[p_2]$ comes from some macro scheme's directive $T[l_1..r_1] \leftarrow T[l_2..r_2]$ such that $p_1 \in  [l_1,r_1]$ and $p_1-l_1 = p_2-l_2$ (if there are multiple directives of this kind, pick any of them). If either $l_1\in [i_1,j_1]$ or $r_1\in [i_1,j_1]$, then $T[i_1,j_1]$ is a primary occurrence and we are done. Otherwise, we set $i_2 = l_2 + (i-l_1)$ and $j_2 = l_2 + (j-l_1)$. By this definition, we have that $T[i_1..j_1] = T[i_2..j_2]$ and $p_2\in [i_2,j_2]$, therefore we can extend our chain to $T[i..j] \leftarrow T[i_2..j_2]$. It is clear that the reasoning can be repeated, yielding that either $T[i_2..j_2]$ is a primary occurrence or our chain can be extended to $T[i..j] \leftarrow T[i_2..j_2] \leftarrow T[i_3..j_3]$ for some substring $T[i_3..j_3]$ such that $p_3 \in [i_3,j_3]$. We repeat the construction for $p_4, p_5, ...$ until either (i) we end up in a chain $T[i..j] \leftarrow ... \leftarrow T[i_k..j_k]$, with $k<k'$, ending in a primary occurrence $T[i_k..j_k]$ of $T[i_1..j_1]$, or (ii) we obtain a chain $T[i_1..j_1] \leftarrow ... \leftarrow T[i_{k'}..j_{k'}]$ such that $p_{k'}\in [i_{k'}, j_{k'}]$ (i.e. we consume all the $p_1,\dots, p_{k'}$). In case (ii), note that $T[p_{k'}] \leftarrow c$ is one of the macro scheme's directives, therefore $T[i_{k'}..j_{k'}]$ is a primary occurrence of $T[i_1..j_1]$. \qed
\end{proof}

Since LZ77 is a particular case of macro scheme, we obtain that it induces a string attractor of size at most $2z$. We can achieve a better bound by exploiting the so-called \emph{primary occurrence} property of LZ77:

\begin{lemma}\label{th: attractor LZ77}
	The Lempel-Ziv factorization of size $z$ of a string $T$ induces a string attractor $\Gamma_{LZ77}$ of size at most $z$.
\end{lemma}
\begin{proof}
	We insert in $\Gamma_{LZ77}$ all positions at the end of a phrase. It is well known (see, e.g.~\cite{kreft2013compressing}) that every text substring has an occurrence crossing a phrase border (these occurrences are usually known as \emph{primary}), therefore we obtain that  $\Gamma_{LZ77}$ is a valid attractor for $T$. \qed
\end{proof}

The run-length Burrows-Wheeler transform seems a completely different paradigm for compressing repetitive strings: while with grammars and macro schemes we copy substrings around in the string, with the RLBWT we build a string permutation by concatenating characters preceding lexicographically-sorted suffixes, and then run-length compress it. This strategy is motivated by the fact that similar substrings (adjacent in the lexicographic order) are often preceded by the same character, therefore the BWT contains long runs of the same letter if the string is repetitive~\cite{makinen2010storage}. We denote with $r$ the number of equal-letter runs in the BWT. The following theorem holds (we use a technique originally developed in~\cite{policritiAlgo17,PP16,Pre16} to get our bound):
%Despite this apparent incompatibility between the RLBWT and the other compression schemes, a closer look shows that the RLBWT is just a particular macro scheme. Run-length compression can be interpreted as copying substrings of length 1 around in the string (in lexicographic order: from the bottom to the top of RLBWT equal-letter runs); moreover, it can be shown that these unary substrings are clustered in $r$ phrases, and sources preserve the adjacency of positions of their relative phrases:  the RLBWT is equivalent to a macro scheme of size $2r$~\cite{gagie2017optimal} (we need $r$ more directives for the $r$ explicitly stored characters) and therefore, for Theorem \ref{th: attractor MS}, it induces a string attractor of size at most $4r$. We can, however, achieve a better bound by using a technique originally developed in~\cite{policritiAlgo17,PP16,Pre16}:

\begin{theorem}\label{th: attractor RLBWT}
	Let $r$ be the number of equal-letter runs in the Burrows-Wheeler transform $BWT(T)$ of a string $T$. $BWT(T)$ induces a string attractor $\Gamma_{RLBWT}$ of size at most $2r$.
\end{theorem}
\begin{proof}
	We insert in $\Gamma_{RLBWT}$ all positions $i$ such that $T[i]$ is the first or last character in its BWT equal-letter run. The size of $\Gamma_{RLBWT}$ is at most $2r$ (less if there are runs of length 1). We now show that $\Gamma_{RLBWT}$ is a valid string attractor.

	Consider any substring $T[i..j]$. We want to show that $T[i..j]$ has an occurrence $T[i'..j'] = T[i..j]$ crossing an element in $\Gamma_{RLBWT}$. We prove the theorem by induction on $j-i$. If $j-i=0$, then $T[i..j]$ is a single character. Since we pick the first and last position of each BWT run, each character appears in at least one attractor position, and the thesis easily follows.
	
	Let $j-i > 0$. By the inductive hypothesis, $T[i+1..j]$ has an occurrence $T[i'+1..j']$ crossing an attractor element $i'+1\leq p \leq j'$. Let $[sp,ep]$ be the $\BWT$ range of $T[i+1..j]$. We distinguish two cases. 
	
	(i) All characters in $\BWT[sp,ep]$ are equal to $T[i]=T[i']$. Then, $T[i'..j']$ is an occurrence of $T[i..j]$ crossing the attractor element $p$.
		
	(ii) $\BWT[sp,ep]$ contains at least one character $c\neq T[i]$. Then, there must be a run of $T[i]$'s ending or beginning in $\BWT[sp,ep]$, meaning that there is a $sp \leq q \leq ep$ such that $\BWT[q] = T[i]$ and the position $p$ corresponding to $q$ belongs to $\Gamma_{RLBWT}$. Then, $T[p..p+(j-i)]$ is an occurrence of $T[i..j]$ crossing the attractor element $p$. \qed

\end{proof}

To conclude, any path-compressed automaton recognizing the string's suffixes also induces a string attractor of the same size. 

\begin{theorem}\label{th: attractor automaton}
	Let $e$ be the number of edges of a path-compressed automaton $\mathcal A$ recognizing all (and only the) substrings of a string $T$. $\mathcal A$ induces a family of string attractors, all of size $e$.
\end{theorem}
\begin{proof}
	We call \emph{root} the starting state of $\mathcal A$.
	Start with empty $\Gamma_{\mathcal A}$. For every edge $(u,v)$ of $\mathcal A$, do the following. Let $T[i..j]$ be any occurrence of the substring read from the root of $\mathcal A$ to the first character in the label of $(u,v)$. We insert $j$ in $\Gamma_{\mathcal A}$.
	
	To see that $\Gamma_{\mathcal A}$ is a valid string attractor of size $e$, consider any substring $T[i..j]$. By definition of $\mathcal A$, $T[i..j]$ defines a path from the root to some node $u$, plus a prefix of the label (possibly, all the characters of the label) of an edge $(u,v)$ exiting from $u$. Let $T[i..k]$, $k\leq j$, be the string read from the root to $u$, plus the first character in the label of $(u,v)$. Then, by definition of $\Gamma_{\mathcal A}$ there is an occurrence $T[i'..k'] = T[i..k]$ such that $k'\in \Gamma_{\mathcal A}$. Since the remaining (possibly empty) suffix $T[k+1..j]$ of $T[i..j]$ ends in the middle of an edge, every occurrence of $T[i..k]$ is followed by $T[k+1..j]$, i.e. $T[i'..i'+(j-i)]$ is an occurrence of $T[i..j]$ crossing the attractor element $k'$. \qed
\end{proof}

The Compact Directed Acyclic Word Graph, of size $e^*$, is the smallest automaton recognizing all string's suffixes.
%: Theorem \ref{th: attractor automaton} implies that the CDAWG  induces a family of string attractors with $e^*$ elements each. As grammars, the CDAWG induces a \emph{family} of string attractors, not just one. Indeed, it can be shown that the CDAWG can be interpreted as the parse tree of a grammar of size $e^*$~\cite{belazzougui2017representing}.

To conclude the section note that the above reductions, in conjunction with Theorems \ref{th:LC} and \ref{th:LR}, give the first known upper-bound to the linguistic complexity $LC(S)$ and lower bound to the length $\ell_{max}$ of the longest repeated substring of a string $S$ as a function of the output size of \emph{any} dictionary compressor.

\subsection{Reductions to compressors}

In this section we derive upper bounds on the approximation rates of some compressors for repetitive strings with respect to the smallest string attractor. We then use these bounds and the reductions of the previous section to uncover new relations between known repetitiveness measures.
The next property follows easily from Definition \ref{def: string attractor} and will be used in the proofs of the following theorems.

\begin{lemma}\label{lemma: superset}
	Any superset of a string attractor is also a string attractor. 
\end{lemma}

We now show that we can derive a bidirectional parse from a string attractor. 

\begin{theorem}\label{th: attractor -> macro}
	Given a string $T\in\Sigma^n$ and a string attractor $\Gamma$ of size $\gamma$ for $T$, we can build a bidirectional parse for $T$ of size $\bigO(\gamma \log (n/\gamma))$.
\end{theorem}

In the next theorem we show how to derive a SLP from a string attractor.

\begin{theorem}\label{th: attractor -> SLP}
	Given a string $T\in\Sigma^n$ and a string attractor $\Gamma$ of size $\gamma$ for $T$, we can build a SLP for $T$ of size $\bigO(\gamma \log^2 (n/\gamma))$.
\end{theorem}

Using the above theorems, we can derive the approximation rates of some compressors for repetitive strings with respect to the smallest string attractor.

\begin{corollary}\label{cor: approximation rates}
	The following bounds hold between the size $g^*$ of the smallest SLP, the size $g_{rl}^*$ of the smallest run-length SLP, the size $z$ of the Lempel-Ziv parse without self-references, 
	the size $b^*$ of the smallest macro scheme, and the size $\gamma^*$ of the smallest string attractor:
	\begin{enumerate}
		\item $b^* \in \bigO(\gamma^* \log (n/\gamma^*))$
		\item $g^*, g^*_{rl}, z \in \bigO(\gamma^* \log^2 (n/\gamma^*))$	
	\end{enumerate}
\end{corollary}
\begin{proof}
	For the first bound, build the bidirectional parse of Theorem \ref{th: attractor -> macro} using a string attractor of minimum size $\gamma^*$. For the second bound, use the same attractor to build the SLP of Theorem \ref{th: attractor -> SLP} and exploit the well-known relation $z\leq g^*$~\cite{rytter2003application}. The result follows from the definitions of smallest macro scheme and smallest SLP. \qed 
\end{proof}

Combining the reductions of Section \ref{sec:compressors->SA} with the results stated in Corollary \ref{cor: approximation rates}, we obtain approximation algorithms for the smallest string attractor. Note that only one of our approximations, however, is computable in polynomial time (unless P=NP): the attractor induced by the Lempel-Ziv factorization. This algorithm computes a $\bigO(\log^2 (n/\gamma^*))$-approximation. In Section \ref{sec:graphs} we show how to compute a $(\ln n)$-approximation of the smallest string attractor in $\bigO(n^2)$ time.

%Note that the results of the previous section imply that all known repetitiveness measures satisfy the lower bound $\Omega(\gamma^*)$.
Combining the construction of Theorem \ref{th: attractor -> SLP} with  the reductions of Theorems \ref{th: attractor g} and \ref{th: attractor MS} we obtain the following new relations between repetitiveness measures: 

\begin{corollary}\label{cor: bounds1}
	The following bounds hold between the size $g^*$ of the smallest SLP, the size $b^*$ of the smallest macro scheme, and the number $r$ of equal-letter runs in the BWT:
	\begin{enumerate}
		\item $g^* \in \bigO(b^*\log^2(n/b^*))$
		\item $g^* \in \bigO(r\log^2(n/r))$ 
	\end{enumerate}
\end{corollary}
\begin{proof}
	For bound 1, build the SLP of Theorem \ref{th: attractor -> SLP} on a string attractor of size $b^*$ induced from the smallest macro scheme (Theorem \ref{th: attractor MS}). Do the same, but using the attractor induced from the RLBWT (Theorem \ref{th: attractor RLBWT}), for bound 2. The results follows from the definition of smallest SLP. \qed 
\end{proof}

%We note that the same bounds have been derived simultaneously by Gagie et al.~\cite{gagie2017optimal} using very different techniques based on locally-consistent parsing. 
Note that, using the well-known relation $z \leq g^*$ for the size $z$ of the Lempel-Ziv parse without self-references~\cite{rytter2003application}, we also obtain:

\begin{corollary}\label{cor: bounds2}
	The following bounds hold between the size $z$ of the Lempel-Ziv parse without self-references, the size $b^*$ of the smallest macro scheme, and the number $r$ of equal-letter runs in the BWT:
	\begin{enumerate}
		\item $z \in \bigO(b^*\log^2(n/b^*))$
		\item $z \in \bigO(r\log^2(n/r))$
	\end{enumerate}
\end{corollary}

The first relation shows that the size of the smallest unidirectional parse without self-references---which can be computed in linear time---is at most a polylogarithmic factor larger than the smallest macro scheme---which is NP-hard to find. The second bound is the first to relate the size of LZ77 with that of the RLBWT \footnote{We are aware that the same bounds have been derived simultaneously by Gagie et al.~\cite{gagie2017optimal} (work not yet published) using very different techniques based on locally-consistent parsing.}, and solves an open problem explicitly stated in~\cite{prezza2016can}.
%(Gagie et al.~\cite{gagie2017optimal} prove also the more strict bounds $z \in \bigO(r\log(n/r))$ and $z \in \bigO(b^*\log(n/b^*))$ for the size of LZ77 with self-references).

\section{Applications: a universal compressed data structure for text extraction}

The problem of efficiently extracting substrings from compressed text representations has lately received a lot of attention in the field of compressed computation. 
On one side, the following lower bound is known to hold for grammar-compressed text~\cite{CVY13}: $\Omega((\log n)^{1-\epsilon}/\log g)$ time is needed to access one random position within 
$O(\mathrm{poly}(g))$ space, for
every constant $\epsilon>0$.
As long as upper-bounds are concerned, several data structures have been proposed in the literature for each distinct compression scheme. In Table \ref{table:extract} we report the best time-space trade-offs known to date, grouped by compression scheme (horizontal lines).
Extracting text from Lempel-Ziv compressed text is a notoriously difficult problem.
No efficient solution exists within $\bigO(z)$ space  (they all require time proportional to the parse's height), although efficient queries can be supported by raising the space by a logarithmic factor~\cite{PhBiCPM17,BGGKOPT15}. Grammars, on the other hand, allow for more compact and time-efficient extraction strategies. 
Bille et al.~\cite{BLRSRW15} have been the first to show how to efficiently perform text extraction within $\bigO(g)$ space. Their time bounds were later improved by Belazzogui et al.\cite{BPT15}, who also showed how to slightly increase the space to $O(g\log^\epsilon n\log(n/g))$   while almost matching the lower bound~\cite{CVY13} in the packed setting. Space-efficient text extraction from the run-length Burrows-Wheeler transform has been an open problem until recently. Standard solutions~\cite{MNSV09} required to spend additional $\bigO(n/s)$ space on top of the RLBWT in order to support extraction in a time proportional to $s$. In a recent publication, Gagie et al.~\cite{gagie2017optimal} showed how to achieve near-optimal extraction time in the packed setting within $O(r\log(n/r))$ space. Belazzogui and Cunial~\cite{BCspire17} showed how to efficiently extract text from a CDAWG-compressed text. Their most recent work~\cite{belazzougui2017representing} shows moreover how to build a grammar of size $\bigO(e)$: this result implies that the same bounds for grammar-compressed text apply to the CDAWG. To conclude, no strategy for efficiently extracting text from general macro schemes is known to date: the only solution we are aware of requires to explicitly navigate the macro scheme's directives, and requires therefore time proportional to the parse's height.

\begin{table}
	\begin{center}
		\begin{tabular}{l|c|c}
			Structure & Space & Extract time \\
			\hline
			Bille et al.~\cite[Lem.~5]{PhBiCPM17}
			& $O(z\log(n/z))$ 
			& $O(\ell+\log(n/z))$ \\
			Gagie et al.~\cite[Thm.~2]{BGGKOPT15}   
			& $O(z\log(n/z))$ 
			& $O((1+\ell/\log_\sigma n)\log(n/z))$ \\
			\hline
			Belazzougui et al.~\cite[Thm.~1]{BPT15} 
			& $O(g)$ 
			& $O(\ell/\log_\sigma n+\log n)$ \\
			Belazzougui et al.~\cite[Thm.~2]{BPT15} 
			& $O(g\log^\epsilon n\log(n/g))$ 
			& $O(\ell/\log_\sigma n + \log n/\log\log n)$ \\
			\hline
			Gagie et al.~\cite[Thm.~2]{gagie2017optimal}
			& $O(r\log(n/r))$
			& $O( \ell\log(\sigma)/w + \log(n/r))$ \\
			\hline
			Belazzougui and Cunial \cite[Thm.~1]{BCspire17}~~
			& $O(e)$ 
			& $O(\ell + \log n)$ \vspace{10pt}\\
		\end{tabular}\caption{Best trade-offs in the literature for extracting  text from compressed representations.}\label{table:extract}
	\end{center}
\end{table}
\vspace{-30pt}
\begin{table}
	\begin{center}
		\begin{tabular}{l|c}
			Space & Extract time \\
			\hline
			$\bigO(\gamma^{1-\epsilon}n^\epsilon)$ & $\bigO(\ell\log(\sigma)/w)$\\
			$\bigO\left(\gamma \log^{\epsilon} n\log(n/\gamma)\right)$ & 
			$\bigO\left(\ell\log(\sigma)/w + \frac{\log(n/\gamma)}{\log\log n}\right)$ \\
			$\bigO\left(\gamma \log^{1+\epsilon}(n/\gamma)\right)$ & 
			$\bigO\left(\ell\log(\sigma)/w + \frac{\log(n/\gamma)}{\log\log (n/\gamma)}\right)$ \\
			$\bigO\left(\gamma \log(n/\gamma)\right)$ & $\bigO(\ell\log(\sigma)/w + \log(n/\gamma))$ \vspace{10pt}
		\end{tabular}\caption{Some interesting trade-offs achievable with the A-DAG of Theorem \ref{thm:extract}, in order of decreasing space and increasing time.}\label{table:extract A-DAG} 
	\end{center}
\end{table}
\vspace{-20pt}

In this section, we show that the powerful abstraction offered by string attractors can be used to efficiently extract text from \emph{any} dictionary-compressed representation. We describe a parametrized data structure---the \emph{Attractor Directed Acyclic Graph}, or A-DAG for short---offering a range of different space-time trade-offs. Our main result is stated in Theorem \ref{thm:extract}. Table \ref{table:extract A-DAG} reports some interesting space-time trade-offs achievable with our data structure. 
To compare the bounds of Tables \ref{table:extract} and \ref{table:extract A-DAG}, just replace $\gamma$ with any of the measures $z,g,r,e,b$ (possible by the reductions described in Section \ref{sec:compressors->SA}).
On one extreme of the parameter space, we achieve optimal extraction time \emph{in the packed setting}. Interestingly, this is achieved within a space---$\bigO(\gamma^{1-\epsilon}n^\epsilon)$ words---that is always $\bigO(n)$ (since $\gamma \leq n$) and is $o(n)$ if the string is asymptotically compressible (i.e. if $\gamma = o(n)$). By the lower bound~\cite{CVY13}, $\Omega((\log n)^{1-\epsilon}/\log g)$ time is needed to access one random position within 
$O(\mathrm{poly}(g))$ space on grammar-compressed strings, for
every constant $\epsilon>0$. Since $g\in \Omega(\log n)$, the space in lines 2 and 3 of Table \ref{table:extract A-DAG} is $O(\mathrm{poly}(g))$ on grammar-induced attractors, and extraction time therefore almost matches the optimum. Moreover, since string attractors are more general than grammars, these trade-offs are close to the optimum also for general string attractors. Note also that the time in line 2 of Table \ref{table:extract A-DAG} is strictly better than that of Line 4 in Table \ref{table:extract} (while the two solutions use the same space) and is optimal in the packed setting for polynomial compression rates, i.e. $n/\gamma \in \mathrm{polylog}(n)$. 
The A-DAG is, in general, faster than all existing extraction strategies for grammars, LZ77, Burrows-Wheeler transform, and CDAWG described in Table \ref{table:extract}. Our solution is, moreover, the first supporting fast extraction on general macro schemes. The A-DAG is a generalization of a data structure proposed very recently by Gagie et al.~\cite{gagie2017optimal} for extracting text compressed with the RLBWT. 

\begin{theorem}\label{thm:extract}
	Let $T[1..n]$ be a string over alphabet $[1..\sigma]$, and let $\Gamma$ be a string attractor of size $\gamma$ for $T$. For any integer parameter $\tau\geq 2$, we can store a data structure---the A-DAG of $\langle T, \Gamma \rangle $---of $\bigO(\gamma \tau \log_\tau(n/\gamma))$ words supporting the extraction of any length-$\ell$ substring of $T$ in $\bigO(\log_\tau(n/\gamma) + \ell\log(\sigma)/w)$ time.
\end{theorem}

\section{Generalizations to labeled graphs}\label{sec:graphs}

In this section we propose a generalization of string attractors to labeled directed graphs capturing all distinct paths (interpreted as strings) in the graph.

In the following, $G=\langle V, E, \delta \rangle$ denotes an edge-labeled directed graph with nodes from $V$, edges from $E \subseteq V\times V$, and labeling function $\delta:E \rightarrow \Sigma$. 

\begin{definition}\label{def:path att}
	A \emph{path attractor} of $G$ is a set of edges $\Gamma \subseteq E$ such that for every path $v_1\dots v_k$ of $G$ there exists a path $v'_1\dots v'_k$ such that $\delta(\langle v_i,v_{i+1}\rangle) = \delta(\langle v'_i,v'_{i+1}\rangle)$ for $i=1,\dots, k-1$ and $\langle v'_j,v'_{j+1}\rangle \in\Gamma$ for some $1\leq j< k$.	
\end{definition}

Note that Definition \ref{def:path att} is equivalent to Definition \ref{def: string attractor} when $G$ is a path graph.
We now study the complexity of determining whether a graph has a path attractor of some fixed size $k$.

\begin{definition} \emph{Path attractor problem} 
	\begin{itemize}
		\item{Input:}\ \ \ \ \ $\langle G, k\rangle$, where $G$ is a labeled graph and $k\geq 1$ is an integer
		\item{Question:} Does $G$ have a path attractor of size at most $k$?
	\end{itemize}
\end{definition}

\begin{lemma}\label{lem:NP-hardness}
	The path attractor problem is NP-hard
\end{lemma}
\begin{proof}
	The idea is to show a reduction from set-cover to the path attractor problem. The reduction is obtained by converting the $t$ sets of the set-cover instance to $t$ tries built using the binary representations of the elements in the sets, and observing that a path attractor on these trees corresponds to a cover of the universe in the set-cover instance (and vice-versa). See Appendix \ref{sec:app} for  full formal proof.
\end{proof}

\begin{lemma}\label{lem:NP}
	The path attractor problem belongs to NP when $G$ is a tree
\end{lemma}
\begin{proof}
	The number of distinct paths in a labeled tree with $n$ nodes is clearly polynomial in $n$, so it takes polynomial time to verify that the edges of a given path attractor intersect at least one occurrence of all distinct paths in the tree.\qed
\end{proof}

Since in Lemma \ref{lem:NP-hardness} the graph $G$ is a labeled tree, Lemmas \ref{lem:NP-hardness} and \ref{lem:NP} imply the following:

\begin{theorem}
	The path attractor problem is NP-complete on labeled trees
\end{theorem}

In the following theorem we show how to compute a logarithmic approximation of the smallest path attractor when the graph is a labeled tree. 
%We note that it is not too hard to compute such an approximation in cubic time by reducing the problem to a set-cover instance, and then using the greedy algorithm described in~\cite{chvatal1979greedy} to find an approximation of this instance. In the following theorem, we show how to reduce the exponent of this polynomial by using the suffix tree of the tree~\cite{breslauer1998suffix}.

\begin{theorem}\label{th: approximation}
	If $G = \langle V, E = \{e_1, \dots, e_n\} \rangle$ is a labeled tree, we can compute a  $(\ln n)$-approximation of the smallest path attractor of $G$ in $\bigO(n^2)$ time. 
\end{theorem}

When $G$ is a string (i.e. a path graph), Theorem \ref{th: approximation} yields a string attractor of size $\gamma^*\ln n$ computable in $\bigO(n^2)$ time. Note that we already provided a logarithmic approximation of the smallest string attractor in the previous section by showing that $b^* \in \bigO(\gamma^* \log (n/\gamma^*))$ (Corollary \ref{cor: approximation rates}). However, finding the smallest macro scheme is a NP-complete~\cite{gallant1982string} problem, so Theorem \ref{th: approximation} is our first polynomial-time algorithm computing such an approximation.

%
%We therefore propose a more powerful extension of string attractors to general labeled graphs. In the following theorem, $H' \cong H$ indicates that graphs $H$ and $H'$ are isomorphic, i.e. there exists a bijection between $H'$'s and $H$'s vertexes preserving edges and labels.
%
%\begin{definition}\label{def:subgraph attractor}
%	A \emph{subgraph attractor} of $G$ is a set of edges $\Gamma \subseteq E$ such that for every connected subgraph $H\subseteq G$ there exists a subgraph $H'\subseteq G$ such that $H' \cong H$ and $\Gamma \cap H' \neq \emptyset$
%\end{definition}
%
%Note that Definition \ref{def:subgraph attractor} is equivalent to Definition \ref{def: string attractor} when $G$ is a path graph. We leave open the problem of determining the complexity of the smallest subgraph attractor problem (although it seems plausible, from the NP-completeness of subgraph isomorphism, that the problem is at least NP-hard).

\section{Conclusions}

We introduced a new combinatorial object---the string attractor---that generalizes all known compressors for repetitive strings and is related with the notion of string's repetitiveness. We proved that the smallest attractor problem becomes NP-complete when generalized to labeled trees, and described several approximation algorithms. In particular, the smallest path attractor on trees (and, in particular, the smallest string attractor) can be approximated within a factor $\ln n$ in $\bigO(n^2)$ time, and the smallest string attractor of size $\gamma^*$ can be approximated within a factor $\bigO(\log^2(n/\gamma^*))$ in $\bigO(n)$ time. We moreover used these combinatorial objects to uncover new relations between repetitiveness measures of string compressors and to obtain a universal compressed data structure for text extraction.

Our paper leaves many exciting open problems. In particular: determining the computational complexity of the smallest string attractor problem, finding better approximation algorithms for the smallest string/path attractor, deriving new relations between repetitiveness measures using the techniques here developed, and deriving better bounds for the linguistic complexity of a string as a function of the size of its smallest attractor.
Finally, note that path attractors capture all sequences of labels in a graph, but do not take into account the graph topology. A possible, more powerful, extension could be the that of considering a set $\Gamma$ of graph's edges such that every distinct connected subgraph has an isomorphic occurrence crossing one of the elements in $\Gamma$.

\subsubsection{Acknowledgements}

I would like to thank Alberto Policriti for many fruitful discussions on the topic. Thanks also to Gonzalo Navarro, Travis Gagie, Philip Bille, Inge Li G\o rtz, Mikko Berggren Ettienne, and Anders Roy Christiansen for many constructive suggestions.

\bibliographystyle{plain}
\bibliography{attractors}

\appendix

\section{Appendix}\label{sec:app}

\subsection{Proof of Lemma \ref{lem: grammars}}

Consider any substring $T[i..j]$. Either the start rule of $G$ is of the form $S\rightarrow A_1B_1$, or $S \rightarrow Z_1^\ell$, for some $\ell >2$. Then, either (i) $T[i..j]$ is fully contained in the strings generated by $A_1, B_1$, or $Z_1$, or (ii) it spans $A_1|B_1$ (in the first case) or $Z_1^{\ell_1}|Z_1^{\ell_2}$, with $\ell_1+\ell_2=\ell$ (in the second case), where we highlighted with a vertical bar the crossing point. 
In case (i), we apply the same reasoning recursively until falling in case (ii). 
In case (ii), if $T[i..j]$ spans $A_i|B_i$, for some $1\leq i \leq g'$, then we have our thesis. The remaining case to consider is the one where $T[i..j]$ spans $Z_i^{\ell_1}|Z_i^{\ell_2}$, for some $1 \leq i \leq g''$ and $\ell_1,\ell_2>1$. This means that $T[i..j]$ can be written as a suffix of $Z_i$ followed by a concatenation of $k\geq 0$ copies of $Z_i$ followed by a prefix of $Z_i$, i.e. $T[i..j] = Z_i[l..|Z_i|]Z_i^{k}Z_i[1..r]$ for some $1 \leq l \leq |Z_i|$, $k\geq 0$, and $0 \leq r \leq |Z_i|$. Then, $T[i..j]$ has also an occurrence overlapping the first occurrence of $Z_i$ in the term $Z_i^{\ell_1+\ell_2}$, i.e. crossing the two factors $Z_i$ and $Z_i^{\ell_1+\ell_2-1}$ of $Z_i^{\ell_1+\ell_2}$: $T[i..j] = Z_i[l..|Z_i|]Z_i^{\ell_1+\ell_2-1}[1..(j-i)-(|Z_i|-l)]$. \qed

\subsection{Proof of Theorem \ref{th: attractor -> macro}}

	We add $\gamma$ equally-spaced attractor elements following Lemma \ref{lemma: superset}. 
	%We first prove the case $\tau = 2$, and then briefly show how to generalize it to $\tau \geq 2$.
	We define phrases of the parse around attractor elements in a ``concentric exponential fashion'', as follows.
	Characters on attractor positions are explicitly stored.
	Let $i_1 < i_2$ be two consecutive attractor elements. Let $m = \lfloor(i_1 + i_2)/2\rfloor$ be the middle position between them. We create a phrase of length 1 in position $i_1+1$, followed by a phrase of length 2, followed by a phrase of length 4, ... until the new phrase does not include position $m$. We do the same (but right-to-left) for position $i_2$. We finally add a phrase in the middle, i.e. in the remaining ``hole'' covering position $m$. For the phrases' sources, we use any of their occurrences crossing an attractor element (possible by definition of $\Gamma$). 
	
	Suppose we are to extract $T[i]$, and $i$ is inside a phrase of length $\leq 2^e$, for some $e$. Let $i'$ be the position from where $T[i]$ is copied according to our bidirectional parse. By the way we defined the scheme, it is not hard to see that $i'$ is either an explicitly stored character or lies inside a phrase of length\footnote{To see this, note that $2^e = 1 + 2^0 + 2^1 + 2^2 + \cdots, + 2^{e-1}$: these are the lengths of phrases following (and preceding) attractor elements (included). In the worst case, position $i'$ falls inside the longest such phrase, of length $2^{e-1}$} $\leq 2^{e-1}$. 
	Since attractor elements are at distance at most $n/\gamma$ from each other, both the parse height and the number of phrases we introduce per attractor element are $\bigO(\log(n/\gamma))$. \qed
			
\subsection{Proof of Theorem \ref{th: attractor -> SLP}}
We first build the bidirectional parse of Theorem \ref{th: attractor -> macro}, with $\bigO(\gamma \log(n/\gamma))$ phrases of length at most $n/\gamma$ each. 
We will process phrases in order of increasing length.  We maintain the following invariant. Every maximal substrings $T[i..j]$ covered by processed phrases is collapsed into a single nonterminal $Y$. Moreover, once we finish processing a phrase $T[i'..j']$, the phrase will be represented by a single nonterminal $X$ (expanding to $T[i'..j']$). We will create $\bigO(1)$ new nonterminals to merge $X$ with the (at most) two adjacent nonterminals representing all contiguous processed phrases to keep our invariant true. It is clear that, once all phrases have been processed, our invariant implies that the entire string is collapsed into a single nonterminal $S$. We now show how to process a phrase and analyze the number of nonterminal introduced by the process.

The overall idea is to map a phrase on its source and copy the source's parse into nonterminals, introducing new nonterminals at the borders if needed. By the bidirectional parse's definition, the source of any phrase $T[i..j]$ overlaps only phrases shorter than $j-i+1$ characters. Since we process phrases in order of increasing length, phrases overlapping the source have already been processed and therefore their parse into nonterminals is well-defined.

We make sure that our string is parsed in levels. 
At level $i>0$, we parse the nonterminals (or terminals if $i=1$) of level $i-1$ in groups of length 2 or 3, and replace each group with a new nonterminal. 

We start by parsing each maximal substring $T[i..j]$ containing only phrases of length 1 into arbitrary blocks of length 2 or 3. We create a new nonterminal per block. Note that, by the way the parse is defined, this is always possible (since $j-i+1 \geq 2$ always holds).
We repeat this process recursively---grouping nonterminals at level $k\geq0$ to form new nonterminals at level $k+1$---until $T[i..j]$ is collapsed into a single nonterminal.
Our invariant now holds for the base case, i.e. for phrases of length $t=1$: each maximal substring containing only phrases of length $\leq t$ is collapsed into a single nonterminal.
We now proceed with phrases of length $\geq 2$, in order of increasing length. Let $T[i..j]$ be a phrase to be processed, with source $T[i'..j']$. By definition of the parse, $T[i'..j']$ overlaps only phrases of length at most $j-i$ and, by inductive hypothesis, these phrases have already been processed. We group characters of $T[i..j]$ in blocks of length 2 or 3 copying the parse of $T[i'..j']$ at level 0. Note that this might not be possible for the borders of length 1 or 2 of $T[i..j]$: this is the case if the block containing $T[i']$ starts before position $i'$ (symmetric for $T[j']$). In this case, we create $\bigO(1)$ new nonterminals as follows. If $T[i'-1,i',i'+1]$ form a block, then we group $T[i,i+1]$ in a block of length 2 and collapse it into a new nonterminal at level 1. If, on the other hand, $T[i'-1,i']$ form a block, we consider two sub-cases. If $T[i'+1,i'+2]$ form a block, then we create the block to $T[i,i+1,i+2]$ and collapse it into a new nonterminal at level 1. If $T[i'+1,i'+2,i'+3]$ form a block, then we create the two blocks $T[i,i+1]$ and $T[i+2,i+3]$ and collapse them into 2 new nonterminals at level 1. We repeat this process for the nonterminals at level $k\geq 1$ that were copied from $T[i'..j']$, grouping them in blocks of length 2 or 3 according to the source and creating $\bigO(1)$ new nonterminals at level $k+1$ to cover the borders. After $\bigO(\log(n/\gamma))$ levels, $T[i..j]$ is collapsed into a single nonterminal. Since we create $\bigO(1)$ new nonterminals per level, we introduce overall $\bigO(\log(n/\gamma))$ new nonterminals. 

At this point, let $Y$ be the nonterminal just created that expands to  $T[i..j]$. Since we process phrases by increasing length, $Y$ is either followed ($YX$), preceded ($XY$), or in the middle ($X_1YX_2$) of one or two nonterminals expanding to a maximal substring containing contiguous processed phrases. We now show how to collapse these two or three nonterminals in order to maintain our invariant true, while at the same time satisfying the property that nonterminals at level $i$ expand to two or three nonterminals at level $i-1$. We show the procedure in the case $Y$ is preceded by a nonterminal $X$, i.e. we want to collapse $XY$ into a single nonterminal. The other two cases can then easily be derived using the same technique. Let $i_X$ and $i_Y$ be the levels of $X$ and $Y$, and let us assume that  $i_X \leq i_Y$ (the case $i_X > i_Y$ is symmetric). If $i_X=i_Y$, then we just create a new nonterminal $W\rightarrow XY$ and we are done. Otherwise, let $Y_1\dots Y_t$, with $t\geq2$, be the sequence of nonterminals that are the expansion of $Y$ at level $i_X$. Our goal is to collapse the sequence $XY_1\dots Y_t$ into a single nonterminal, while introducing at most $\bigO(\log(n/\gamma))$ new nonterminals (that are charged to the phrase $T[i..j]$: overall, we will therefore introduce $\bigO(\log(n/\gamma))$ new nonterminals per phrase). The parsing of $Y_1\dots Y_t$ into blocks is already defined (by the expansion of $Y$), so we only need to copy it while adjusting the left border in order to include $X$. We distinguish two cases. If $Y_1$ and $Y_2$ are grouped into a single block, then we replace this block with the new block $XY_1Y_2$ and collapse it in a new nonterminal at level $i_X+1$. If, on the other hand, $Y_1$, $Y_2$, and $Y_3$ are grouped into a single block then we replace it with the two blocks $XY_1$ and $Y_2Y_3$ and collapse them in two new nonterminals at level  $i_X+1$. We repeat the same procedure at levels  $i_X+1, i_X+2, \dots, i_Y$, until everything is collapsed in a single nonterminal. At each level we introduce one or two new nonterminals, therefore overall we introduce at most $2(i_Y-i_X)+1 \in \bigO(\log(n/\gamma))$ new nonterminals.\qed

\subsection{Proof of Theorem \ref{thm:extract}}

We describe a data structure supporting the extraction of $\alpha = \frac{w\log_\tau(n/\gamma)}{\log\sigma}$ packed characters in $O(\log_\tau(n/\gamma))$ time. To extract a substring of length $\ell$ we divide it into $\lceil\ell/\alpha\rceil$ blocks and extract each block with the proposed data structure. Overall, this will take $O((\ell/\alpha+1)\log_\tau(n/\gamma)) = O(\log_\tau(n/\gamma) + \ell\log(\sigma)/w)$ time. 

Our data structure is stored in $O(\log_\tau(n/\gamma))$ levels. For simplicity, we assume that $\gamma$ divides $n$ and that $n/\gamma$ is a power of $2\tau$. Intuitively, we will build a DAG with nodes of out-degree $2\tau$; each node will be associated with a substring whose length is exponentially decreasing in the levels (with base $2\tau$).

The top level (level 0) is special: we divide the string into $\gamma$ blocks 	$T[1..n/\gamma]\,T[n/\gamma+1..2n/\gamma]\dots T[n-n/\gamma+1..n]$ of size $n/\gamma$. For levels $i>0$, we let $s_i = n/(\gamma\cdot \tau^i)$ and, for every element $j\in\Gamma$, we consider the $2\tau$ non-overlapping blocks of length $s_i$: 
$T[j-s_i\cdot k +1 ... j-s_i\cdot (k-1)]$ and $T[j+s_i\cdot (k-1) + 1 ... j+s_i\cdot k]$, for $k=1,\dots, \tau$. 
Each such block is composed of two half-blocks of length $s_i/2$. In total, there are $4\tau$ half blocks. 
We moreover consider a sequence of $4\tau-1$ additional consecutive and non-overlapping half-blocks of length $s_i/2$, starting in the middle of the first half-block above defined and ending in the middle of the last.
Note that, with this choice of blocks, at level $i$ for any substring $S$ of length at most $s_i/4$ (inside the considered regions of length $2\tau\cdot s_i$ around elements of $\Gamma$) we can always find a half-block fully containing $S$. This property will now be used to map ``short'' strings from the first to last level of our structure without splitting them, until reaching explicitly stored characters at some level that we define below.

From the definition of string attractor, blocks at level $0$ and each half-block at level $i> 0$ have an occurrence at level $i+1$ crossing some position in $\Gamma$. Such an occurrence can be fully identified by the coordinate $\langle \mathit{off}, j\rangle$, for $0 \leq \mathit{off} < 2s_i$ and $j\in\Gamma$, indicating that the occurrence starts at position $j-s_i + \mathit{off} + 1$.
Let $i^*$ be the smallest number such that $s_{i^*} < 4\alpha = 
\frac{4w\log_\tau(n/\gamma)}{\log\sigma}$. Then $i^*$ is the last level of our structure.
At this level, we explicitly store a packed string with the characters of the blocks. This uses in total $O(\gamma \cdot \tau \cdot  s_{i^*}\log(\sigma)/w) = O(\gamma\tau\log_\tau(n/\gamma))$ words of space. All the blocks at level 0 and half-block at levels $0<i<i^*$ store instead the coordinates $\langle \mathit{off},j\rangle$ of their primary occurrence in the next level. At level $i^*-1$, these coordinates point inside the strings of explicitly stored characters.

Let $S= T[i..i+\alpha-1]$ be the substring to be extracted. Note that we can assume $n/\gamma \geq \alpha$; otherwise all the string can be stored in plain packed form using $n\log(\sigma)/w < \alpha \gamma\log(\sigma)/w \in O(\gamma\log_\tau (n/\gamma))$ words and we do not need any data structure. It follows that $S$ either spans two blocks at level 0, or it is contained in a single block. The former case can be solved with two queries of the latter, so we assume, without losing generality, that $S$ is fully contained inside a block at level $0$. To retrieve $S$, we map it down to the next levels (using the stored coordinates of primary occurrences of half-blocks) as a contiguous substring as long as this is possible, that is, as long as it fits inside a single half-block. Note that, thanks to the way half-blocks overlap, this is always possible as long as $\alpha \leq s_i/4$. By definition, then, we arrive in this way precisely to level $i^*$, where characters are stored explicitly and we can return the packed substring. Note also that, since blocks in the same level have the same length, at each level we spend only constant time to find the pointer to the next level.
\qed

\subsection{Proof of Lemma \ref{lem:NP-hardness}}

	We show a reduction from \texttt{set-cover}. Let $\langle U, S_1, \dots, S_t \rangle$ be an instance of \texttt{set-cover}, where $U=\{u_1, \dots, u_n\}$ is the universe and $\bigcup_{i=1}^{t} S_{i} = U$. Without loss of generality, we can assume $u_i = i-1$, for $i=1, \dots, n$, and $n=2^m$ for some $m\geq0$: if $n$ is not a power of two, then we add a new set $S_{t+1} = [|U|,2^{\lceil \log_2|U| \rceil}-1]$ and replace the universe by $U' = U\ \cup\ S_{t+1}$ (i.e. we add new elements until reaching the next power of two). It is clear that $S_{i_1}, \dots, S_{i_q}$ cover $U$ if and only if $S_{i_1}, \dots, S_{i_q}, S_{t+1}$ cover $U'$, so the problem with $n$ being a power of two is still NP-complete.

We build a labeled tree $\mathcal T$ as follows. The alphabet is $\langle 0,1, s_1, \dots, s_t\rangle$. The root $r(\mathcal T)$ has $t$ children $c_1, \dots, c_t$ with incoming edges labeled, respectively, $s_1, \dots, s_t$. Each $c_i$ has two children nodes, $c_i^0$ and $c_i^1$. Edges $\langle c_i, c_i^0\rangle$ and $\langle c_i, c_i^1\rangle$ are labeled 0 and 1, respectively.
Let $S_i = \{u_{j_{i,1}}, \dots, u_{j_{i,b_i}}\}$. The subtrees rooted in $c_i^0$ and $c_i^1$ are identical: each is a binary trie containing the $m$-digits binary representations of $u_{j_{i,1}}, \dots, u_{j_{i,b_i}}$. We claim that our \texttt{set-cover} instance $\langle U, S_1, \dots, S_t \rangle$ has a solution of size at most $\gamma$ if and only if $\mathcal T$ has a path attractor of size at most $t+2\gamma$. 

$(\Rightarrow)$ Let $S_{i_1}, \dots, S_{i_q}$ be a cover of $U$. We want to prove that $\mathcal T$ has an attractor of size at most $t+2q$.
We claim that $\Gamma = \{ \langle r(\mathcal T), c_i \rangle: i=1, \dots, t  \}\ \cup\ \{ \langle c_i, c_i^0 \rangle : i=i_1, \dots, i_q \} \cup\ \{ \langle c_i, c_i^1 \rangle : i=i_1, \dots, i_q \}$ is such a path attractor for $\mathcal T$. By the way $\mathcal T$ is defined, the strings labeling paths in this tree either start with a $s_i$, for $1\leq i \leq t$, or are any combination of at most $m+1$ bits (i.e. characters in $\{0,1\}$). In particular, note that \emph{all} combinations of at most $m+1$ bits occur as paths, since (i) $|U|=2^m$, (ii) $\bigcup_{i=1}^{t} S_{i} = U$, and (iii) the binary representations of $S_i$-elements are always preceded by the labels 0 and 1 in the tree's paths. All paths whose labels start with a $s_i$ are clearly captured by the edges $\{ \langle r(\mathcal T), c_i \rangle: i=1, \dots, t  \}\in \Gamma$. We now focus on the binary strings labeling the remaining paths.
Since $S_{i_1}, \dots, S_{i_q}$ is a cover of $U$, the binary representations of elements in $S_{i_1}, \dots, S_{i_q}$ are all possible combinations of $m$ bits. This implies that the set of labels of paths starting from edges in $\{ \langle c_i, c_i^x \rangle : i=i_1, \dots, i_q \}$, with $x\in\{0,1\}$, are all possible combinations of $x$ followed by any sequence of $m$ bits. Since in $\Gamma$ we include both $\langle c_i, c_i^0 \rangle$ and $\langle c_i, c_i^1 \rangle$ for $i=i_1, \dots, i_q$, these edges capture all prefixes of paths labeled with any combination of $m+1$ bits. Since the closure under prefix of the set $\{0,1\}^{m+1}$ is exactly $\{0,1\}^{\leq m+1}$, we finally obtain that edges in $\Gamma$ capture all paths in $\mathcal T$.

$(\Leftarrow)$ Let $\Gamma$ be a path attractor for $\mathcal T$, of size $q$. We want to prove that our \texttt{set-cover} instance has a solution of size at most $(q - t)/2$.	
First, note that $\{\langle r(\mathcal T), c_i \rangle: i=1, \dots, t\} \subseteq \Gamma$, as otherwise the $t$ unary paths labeled $s_1, \dots, s_t$ would not be represented by the attractor. Let $\alpha = q-t$. The $\alpha$ edges in $\Gamma -   \{\langle r(\mathcal T), c_i \rangle: i=1, \dots, t\}$ intersect all possible sequences of at most $m+1$ bits (by the attractor's definition). We create a cover of $U$ of the desired size as follows. We first initialize two empty sets $R_0$ and $R_1$. Then, for every edge $\langle u, v\rangle$ in $\Gamma -   \{\langle r(\mathcal T), c_i \rangle: i=1, \dots, t\}$, if $\langle u, v\rangle = \langle c_i, c_i^x\rangle$ or $\langle u, v\rangle$ belongs to the subtree rooted in $c_i^x$, for some $1\leq i \leq t$ and $x\in \{0,1\}$, then we insert $S_i$ in $R_x$. We claim that both $R_0$ and $R_1$ are a cover of $U$. In particular, note that $|R_0|+|R_1|\leq \alpha$: we pick as solution of our \texttt{set-cover} instance the $R_x$ of minimum size $|R_x|\leq \alpha/2 = (q - t)/2$, that is, our claim. We now prove that $R_1$ is a cover of $U$ (the proof is analogous for $R_0$). 
Note that the following two facts hold. (i)
By definition, the tree rooted in $c_i^1$ is a trie containing the binary representations of elements from $S_i$. (ii) Edges $\langle u, v\rangle$ in $\Gamma -   \{\langle r(\mathcal T), c_i \rangle: i=1, \dots, t\}$ such that $\langle u, v\rangle = \langle c_i, c_i^1\rangle$ or $\langle u, v\rangle$ belongs to the subtree rooted in some $c_i^1$ capture all paths labeled with $1\{0,1\}^m$ (from the attractor's definition). From (i) and (ii) it is easy to derive that sets in $R_1$ must cover the whole $U$. Suppose, by contradiction, that some $u\in U$ is not covered by $R_1$, and let $B_u$ be the binary representation of $u$. Let moreover $S_{i_1}, \dots, S_{i_k}$ be all the sets containing $u$. Since $u\in U$ is not covered by $R_1$, none of the $S_{i_1}, \dots, S_{i_k}$ belongs to $R_1$. Then, by definition of $R_1$, this implies that no edge of $\Gamma$ is equal to $\langle c_i, c_i^x\rangle$ or belongs to the subtree rooted in $c_i^x$, for any $i\in \{i_1, \dots, i_k\}$. Since these are the only edges potentially intersecting path $1B_u$, this means that $1B_u$ is not captured by $\Gamma$, which is a contradiction.	\qed

\subsection{Proof of Theorem \ref{th: approximation}}

	We show a reduction to set-cover, and then solve the instance with the greedy algorithm~\cite{chvatal1979greedy} that at each round picks the set covering the largest number of non-covered elements in the universe. It is known~\cite{chvatal1979greedy} that the greedy algorithm computes a $(\ln n)$-approximation of the optimal solution.
We make use of the suffix tree $\mathcal T(G)$ of the tree $G$~\cite{breslauer1998suffix}. $\mathcal T(G)$ is obtained by conceptually building the trie containing all reversed prefixes of root-to-leaves paths in $G$, and then applying path compression (as in the standard suffix tree of a string). $\mathcal T(G)$ takes $\bigO(n)$ space, can be built in $\bigO(n)$ time, and contains a node (implicit or explicit) for each reversed path in $G$~\cite{breslauer1998suffix} (i.e. all $G$'s reversed paths are represented from the root to some node of $\mathcal T(G)$).

The universe of our set-cover instance is composed of all $N\in \bigO(n)$ edges of $\mathcal T(G)$: $U = E' = \{e'_1, e'_2, \dots, e'_N\}$. The collection of sets is $S = \{s_1, \dots, s_n\}$: one set per $G$'s edge. 
Let $st(e')$, where $e'$ is a $\mathcal T(G)$'s edge, denote the string read from the root of $\mathcal T(G)$ to the first character (included) in the label of $e'$.
Set $s_i$ contains all $\mathcal T(G)$'s edges $e'_j$ such that there is an occurrence of $\overleftarrow{st(e'_j)}$ in $G$ (i.e. a path labeled $\overleftarrow{st(e'_j)}$) crossing  edge $e_i$. 
We claim that $S=\{s_{i_1}, \dots, s_{i_t}\}$ is a solution to this instance of set-cover if and only if $\Gamma=\{e_{i_1}, \dots, e_{i_{t}}\}$ is a path attractor for $G$. In particular, note that the one-to-one correspondence implies that a $c$-approximation to the set-cover instance yields a $c$-approximation of the smallest path attractor.

$(\Rightarrow)$ Let $S = \{s_{i_1}, \dots, s_{i_t}\}$ be a solution to the set-cover instance. We want to show that $\Gamma = \{e_{i_1}, \dots, e_{i_{t}}\} \subseteq E$ is a path attractor for $G$.
Now, consider the label $s$ of any path in $G$. By definition of $\mathcal T(G)$, $\overleftarrow s$ is a path starting from the root and ending inside the label of some edge $e'$ of $\mathcal T(G)$. In particular, we can write $s = \overleftarrow{st(e')x}$, for some prefix $\overleftarrow x$ of $s$.	
Since $S$ is a solution to the set-cover instance, there exists a $s_i\in S$ such that $e'\in s_i$. By the way the set-cover instance is defined, this implies that $e_i\in \Gamma$ crosses an occurrence of $\overleftarrow{st(e')}$ in $G$. But then, the fact that $x$ is the continuation of $st(e')$ inside the label of $e'$ means that all paths labeled $\overleftarrow{st(e')}$ in $G$ are always preceded by $\overleftarrow x$, i.e. $e_i\in \Gamma$ crosses an occurrence of $\overleftarrow{st(e')x} = s$ in $G$. Since $s$ is a generic label of a path in $G$, we obtain that $\Gamma$ is a valid path attractor.

$(\Leftarrow)$ Let $\Gamma = \{e_{i_1}, \dots, e_{i_{t}}\} \subseteq E$ be a path attractor for $G$. We want to show that $S = \{s_{i_1}, \dots, s_{i_t}\}$ is a solution to the set-cover instance. Pick any suffix tree edge $e'\in U$. Now, consider the string $st(e')$. By the definition of $\mathcal T(G)$, $\overleftarrow{st(e')}$ is a path in $G$. Since $\Gamma$ is an attractor for $G$, there exist a $e_i\in \Gamma$ intersecting an occurrence of $\overleftarrow{st(e')}$. Then, by definition of the set-cover instance, $e' \in s_i \in S$. Since $e'$ is a generic element of $U$, we obtain that $S$ is a cover of $U$.

We implement the logarithmic approximation algorithm as follows. We maintain an undirected bipartite graph $H=\langle E,E', F \subseteq E\times E' \rangle$, having as vertexes all elements in $E$ (i.e. $G$'s edges) and $E'$ (i.e. $\mathcal T(G)$'s edges). We add an edge $\langle e,e' \rangle$ in $F$ whenever $G$'s edge $e$ crosses an occurrence of $st(e')$. We keep $G$'s edges in an array $P$, and associate to each of them an integer storing their degree in $H$: $deg(e) = |\{\langle e,e' \rangle \in F\}|$. The degrees can be computed in $|F|\in\bigO(n^2)$ time at the beginning. 
Our algorithm runs in $\bigO(n)$ steps: at each step, we pick the $e\in P$ with largest $deg(e)$. After that, for each $e'$ such that $\langle e,e' \rangle \in F$, we decrease by one the degree of neighbors of $e'$ in $H$, i.e. of all $e_i\in E$ such that  $\langle e_i, e' \rangle \in F$, and we remove $e$, all its adjacent edges, and all its neighbors from $H$. We stop the procedure when $H$ does not have anymore vertexes from $E'$. 

Since each edge of $H$ can be removed only once, re-computing the degrees takes overall $\bigO(n^2)$ time. Picking the minimum from $P$ takes $\bigO(n)$ time at each step. Overall, the procedure runs in $\bigO(n^2)$ time. \qed

\end{document}